\documentclass[review]{elsarticle}

\usepackage{lineno,hyperref}
\usepackage{graphicx}
\usepackage{amssymb}
\usepackage{amsfonts}
\usepackage{subfig}
\usepackage{amssymb}
\usepackage{float}
\usepackage{amsmath,amsthm}
\usepackage{mathptmx}
\usepackage{float}
\usepackage{mathrsfs}
\usepackage{graphicx}
\usepackage{leftidx}
\usepackage{extarrows}
\usepackage{xcolor}
\usepackage{booktabs}
\usepackage{threeparttable}
\usepackage{latexsym}

\newtheorem{assumption}{Assumption}
\newtheorem{defi}{Definition}
\newtheorem{rmk}{Remark}
\newtheorem{thm}{Theorem}

\newtheorem{lem}{Lemma}

\modulolinenumbers[5]

\journal{ ISA Transactions }

\begin{document}

\begin{frontmatter}

\title{A fully distributed event-triggered communication strategy for second-order multi-agent systems consensus}


\author[mymainaddress]{Tao Li}

\author[mymainaddress]{Quan Qiu\corref{mycorrespondingauthor}}
\cortext[mycorrespondingauthor]{Corresponding author}

\author[SecAddress]{Chunjiang Zhao\corref{mycorrespondingauthor}}

\address[mymainaddress]{Beijing Research Center of Intelligent Equipment for Agriculture, Beijing, China}
\address[SecAddress]{National Engineering Research Center for Information Technology in Agriculture, Beijing, China}

\begin{abstract}
	This paper investigates the communication strategy for second-order multi-agent systems with nonlinear dynamics. To save the scarce resources of communication channels, a novel event-triggered communication mechanism is designed without using continuous signals among the followers. To get rid of the centralized information depending on the spectrum of the Laplacian matrix, a consensus protocol with updated coupling gains and an event-triggered strategy of interagent communication with updated thresholds are presented. 
	To minimize the impacts of noise, only relative positions among agents are employed in the protocol.
	With the proposed event-triggered mechanism and the control protocol, the leader-following consensus of MASs and no Zeno behavior included are mathematically proven. The results are verified through numerical examples. 
\end{abstract}

\end{frontmatter}


\section{Introduction}

Over the last decade, ever-increasing research trends concentrate on the studies of multi-agent systems (MASs). A fundamental problem of MASs is designing a networking protocol of consensus, which means that agents converge to a common point or state value. Consensus has been extensively investigated in literature \cite{shang2020resilient,huang2020distributed}. Reviewing the existing studies of this area, one may note that the effect of network-induced communicating constraints over the consensus control performance attracts extensive attention, such as the problems of communication delays \cite{wang2020distributed}, switching topology \cite{Xie2015Event}, discrete interagent information exchange \cite{zhang2020semi}. 
A main source of communication constraints\cite{yan2020performance} is the scarce network bandwidth. Generally speaking, the communication channels of MASs are usually multipurpose and various kinds of interagent information shares the common channels. To achieve desired timeliness, with limited bandwidth, reducing the burden of communication is expected.

It is well-known that the continuous signals among agents usually requires more communication bandwidth than discrete signals. In light of this, the sample-based communicating mechanism is proposed to take samples of exchange of information among agents. Under the control of a well-designed sample-based communicating mechanism, consensus of MASs can be guaranteed with less channel occupancy. In the sample-based mechanisms, there are two kinds which are event-triggering mechanism (ETM) and time-trigger mechanism (TTM). 
TTM takes the samples according to time, which is widely used in many existing sample-based network protocols \cite{Xian2015Event,Zhang2016Survey}. But one potential drawback is that the time-driven sampling, which is independent of system states, feedbacks, communication resources et al., may result in unnecessary and redundant sampled-data\cite{8959085, Hetel2017Recent}. To optimize the strategies of sampling, ETM is proposed, where the sampling is triggered by the predefined conditions. The conditions depend on the system state, the feedback signals or other artificially designed conditions. 

ETM can be dated back to the late 50s in the 20th-century \cite{Ellis1959Extension} and then event-trigger mechanism has been extensively developed during the past decades in control community. One of pioneering work in the research of MAS consensus with ETM is presented by \cite{Tabuada2007Event}, where the control actuation is triggered whenever the errors become large enough w.r.t. the norm of the state. The proposed ETM guarantees the performance of consensus and relaxes the requirements of periodic execution. 
Following \cite{Tabuada2007Event}, some publications also show that ETMs are suitable for a class of first-order MASs \cite{Dimarogonas2012Distributed}, for the double-integrator MASs \cite{Seyboth2013Event}, for a class of linear time-invariant MASs \cite{Guinaldo2012Distributed}.
By ETM, the above literature focuses on reducing updates of the controllers in agents but still requires the continuous interagent exchange. Aiming at alleviating the burden of network channels, \cite{Fan2013Technical} presents an event-triggered algorithm to get rid of the continuous measurements of neighbors' state. Additionally, \cite{Zhu2014Event}  proposes two event-triggered condition functions: one for reducing control updates, the other for avoiding continuous communication among agents. Considering the second-order leader-following MASs with nonlinear dynamic behaviors,   \cite{Li2015Event,Zhao2017Event} propose distributed event-triggered sampling control approaches, where the agents only broadcast their discrete state values and the local controllers are only updated their outputs when the triggering conditions are satisfied.   
\cite{Guo2014A} uses an event-triggered framework which is free of continuous measurement of the triggered condition and gives the sufficient conditions on the consensus of MASs with the linear agents. 
There is a common flaw in the aforementioned work. Centralized information depending on the spectrum of the Laplacian matrix is required a prior in the course of designing ETM and the control protocol. See Remark 2 in \cite{Zhu2014Event}, Eq. (11) in \cite{Zhao2017Event}, Theorem 2 in \cite{Guo2014A}, etc., just to name a few.
It means the ETMs and control protocols are based on the assumption that each agent knows the overall communication topology. Relaxing such a centralized assumption therefore is worthwhile to be investigated to take full advantage of the power of distributed protocols. 
Fortunately, in the studies of fully distributed consensus protocol in MASs, there has been much progress \cite{Yu2011Distributed,Zhongkui2013Distributed,Zhongkui2015Designing, yan2019event,Yuezu2017}. Depending on only local information of each agent, some distributed adaptive consensus protocols are presented and applied to undirected communication graphs \cite{Yu2011Distributed,Zhongkui2013Distributed, goebel2020unifying} and directed communication graphs \cite{Zhongkui2015Designing, Yuezu2017, fu2019robust}. However, these protocols require the continuously local states and interacting states, which does not suit the case that resources of communication and computation are limited. In particular, when combining with ETMs, it is a promising topic to obtain less conservation in terms of reducing the frequency of interagent exchanges.

To the best of the authors’ knowledge, however, the issue is still not appropriately addressed in literature. Designing a novel distributed ETM consensus protocol 
naturally becomes a motivation of the present paper. In this paper, the consensus protocol is expected to be free of centralized information depending on the spectra of the Laplacian matrix. Moreover, agents in the considered MAS are with nonlinear dynamics under the connected undirected graph. The primary contributions are summarized as follows.
a) To overcome the challenging problem that information depending on the spectra of Laplacian matrices is required a prior to select parameters of event-triggered functions, a novel event-triggered sampled-data mechanism with an adaptive threshold is first proposed 
b) The fully distributed consensus protocol for second-order MASs with nonlinear dynamics is designed, which is based on event-triggered sampled-data interacting information among agents. 
c) Only the relative discrete position information is employed in both the event-triggered rule and the consensus protocol, which results in that the undesired velocity measurements can be avoided.

Throughout this paper, $\mathbb{R}^n$ and $\mathbb{R}^{n\times n}$ denote the $n-$Euclidean space and the set of all $n\times n$ real matrices, respectively; $\|\cdot\|$ stands for either the Euclidean vector norm or the spectral norm of a matrix; $\otimes$ denotes the Kronecker product; $I_n$ represents an $n \times n$ identity matrix; $\lambda_{min}(\cdot)$ and $\lambda_{max}$ denote the minimum and maximum eigenvalue of a matrix; $diag\{d_1,\ldots,d_n\}$ denotes the diagonal matrix with the elements $d_1,\ldots,d_n$ on the diagonal.

\section{Preliminaries}\label{2secPre}

The following lemmas are necessary for the analysis of this paper. 
\subsection{Some supporting lemmas}

\begin{lem}\label{lem3}
	Let $\omega: R \rightarrow R$ be a uniformly continuous function on $[0,\infty \} $. Suppose that  $lim_{t \rightarrow \infty} \int_{0}^{t} \omega (\tau) d \tau$ exists and  is finite. Then,
	$$
	\omega(t) \rightarrow 0~~ as ~~t \rightarrow \infty.
	$$
\end{lem}
\begin{lem}\label{lem-SCHUR}
	The following linear symmetric matrix inequality (LMI)
	\begin{equation*}
	S=S^T=
	\begin{pmatrix}
	A&B\\ *&C
	\end{pmatrix}<0,
	\end{equation*}
	is equivalent to one of the following conditions: 
	\begin{enumerate}
		\item  $S<0$;
		\item  $A<0,C-B^TA^{-1}B<0$;
		\item  $C<0,A-BC^{-1}B^T<0$.
	\end{enumerate}
\end{lem}
\begin{lem}\label{lem-BABA} \cite{Slotine2004Applied}
	For the function $V(x,t)$, $V(x,t) \rightarrow 0$ as $t\rightarrow \infty$ holds, when the following conditions is met.
	\begin{enumerate}
		\item $V(x,t)$ has lower bound;
		\item $\dot V(x,t)$ is negative semi-definite;
		\item $\dot V(x,t)$ is a uniformly continuous function w.r.t. time, in other word, $\ddot V(x,t)$ has bound. 
	\end{enumerate}
\end{lem}
\begin{lem}\label{LEM - L}
	The Laplacian matrix $L$ of an undirected graph $\mathcal{G}$ is semi-positive definite, which has a simple zero eigenvalue and all the other eigenvalues are positive if and only if the undirected graph $\mathcal{G}$ is connected.
\end{lem}
\begin{lem}\label{lem - H}\cite{Chen2007Pinning}
	If $L$ is reduciable, $L_{ij} = L_{ji} \leq 0$ for $i \ne j  $, and $\sum_{j=1}^{N}L_{ij} = 0,i=1,2,\ldots,N$ then, for any constant $\varpi >0$, all eigenvalues of the matrix $H= L+B$ are positive, i.e. $\lambda(H)>0$, where $B = diag\{\varpi,0,\ldots,0\}$ 
\end{lem}
\subsection{Graph theory}\label{GT}
The notation of communication graph in this paper are extensively used in literature. 
The networking topology among $N$ follower is modeled by a positively weighted undirected graphs $\mathcal{G}=(\mathcal{V},\mathcal{E},\mathcal{W})$, where $\mathcal{V}$ denotes a nonempty vertex set $\left\lbrace v_1,v_2,\ldots,v_N \right\rbrace $ describing agents; $\mathcal{E} \subset \mathcal{V} \times \mathcal{V}$ denotes the set of undirected edges $e_{ij}$ describing the information exchanging and $\mathcal{W}= ({w}_{ij})_{N\times N}$ denotes the underlying weighted adjacency matrix with nonnegative elements. 
An undirected edge $e_{ij}$ in graph $\mathcal{G}$ means that nodes $v_i$ and $v_j$ can exchange information with each other. If $e_{ij}$ exists between two nodes, $w_{ij} = w_{ji}>0$; otherwise, $w_{ij} = w_{ji}=0$. 
A graph is connected if every vertex in $\mathcal{V}$ is globally reachable and a vertex $i \in \mathcal{V} $ is globally reachable if any vertex other than $i$ has at least one path starting at the vertex and ending at the vertex $i$.
Furthermore, we assume that $i \not\in \mathcal{N}_i$ (no self-loop contained), and hence for all $i \in \mathcal{V}$, $w_{ii}=0$. The Laplacian matrix $\textbf{L}=[l_{ij}]$ is defined by
\begin{equation*}
l_{ij} = \left\{\begin{aligned}
&  \sum_{j\in \mathcal{N}_i}w_{ij}~~i=j,\\
&  -w_{ij}~~~i\neq j.\\
\end{aligned}\right.
\end{equation*}

For the networking topology with a leader, the total communication topology between the leader and its followers can be formulated by graphs $\bar{\mathcal{G}}$, namely, $\mathcal{G} \subset \bar{\mathcal{G}}$. In $\bar{\mathcal{G}}$, one leader can only send information to out-neighboring followers but not receive reversely. Let $K=[k_1,\ldots,k_N]^T$ denote the set of the weights from the leader to its followers. Accordingly, the Laplacian matrix of $\bar{\mathcal{G}}$ is defined by 

\begin{equation*} 
\bar L =\begin{pmatrix}
0&\textbf{0}\\
-K& H
\end{pmatrix},
\end{equation*}
where $H=\left\lbrace h_{ij}\right\rbrace_{N\times N} = \textbf{L}+D$ and $D=diag\{k_1,\ldots,k_N\}$.

\subsection{Problem formulation}
The second-order MAS considered in this paper consists of one leader and $N$ followers, which can be formulated by 
\begin{equation}\label{sys1}
\begin{split}
\dot x_i (t)  =& v_i (t) ,\\
\dot v_i (t)  =& f\left(t,x_i(t),v_i(t) \right) + u_i (t),\\
\end{split}
\end{equation}
where $x_i(t)$,$v_i(t),u_i(t) \in \mathbb{R}^{n}$ denote the position, velocity and control input of the agent $i$, respectively; $f\left(\cdot \right) $ is a continuously differentiable vector-valued nonlinear function to describe the self-dynamics of agents. The dynamics of the leader is governed by
\begin{equation}\label{sys-leader}
\begin{split}
\dot x_0 (t)  =& v_0 (t) ,\\
\dot v_0 (t)  =& f\left(t,x_0(t),v_0(t) \right),\\
\end{split}
\end{equation}
where $x_0$, $v_0$ are the position and velocity of the leader. 
Throughout this paper, the following assumption is made. 
\begin{assumption}\label{ASS1} 
	For the nonlinear function $f(t,x_i(t),v_i(t))$, the velocity state ${v}_i(t)$ is linearly coupled, which means 
	\begin{equation*}
	f(t,x_i(t),v_i(t))\leq \varsigma v_i(t)+ f(t, x_i(t)) ~~~\forall x_i, v_i \in \mathbb{R}^n ,
	\end{equation*} 
	where $\varsigma$ is a scalar or a matrix with proper dimensions. Additionally, for any $ x,y,z,v \in \mathbb{R}^n $, there exists nonnegative constant $\rho$ such that
	\begin{equation*}
	\|f(t,x_i) - f(t,x_{j})\| \leq \rho \|x_i - x_j \|.
	\end{equation*}
\end{assumption}

In existing literature, the event-triggered controller for agent $i$ is usually designed as (taking \cite{Zhao2017Event} as an example)
\begin{equation}\label{existing}
\begin{split}
u_i(t) =&  - \tilde\alpha \sum_{j=1}^{N}w_{ij} \left[  {x}_j (t_k^j)-{x}_i (t_k^i)+ {v}_j (t_k^j)-{v}_i (t_k^i) \right]\\ 
& -\tilde\alpha k_i  \left[ {x}_i (t) - {x}_0 (t) + {v}_i (t) - {v}_0 (t)\right], t\in [t_k^i,t_{k+1}^i ),
\end{split}
\end{equation}
where $\tilde{\alpha}>0$ is coupling strength and $t_k^j \triangleq \text{arg~min}_p \{t-t_p^j|t\geq t_p^j,p\in \mathbb{N}\}$, i.e., $t_k^j$ is the latest triggering time of agent j before time $t$. The control protocol is distributed since each agent only uses local information of neighboring agents, which can be clearly seen in \eqref{existing}. Similar distributed protocols can be found in \cite{Xie2015Event,Zhu2014Event,Li2015Event,Guo2014A}. In these literature, the feasibilities of the consensus criteria depend on that the coupling gains and the eigenvalue of a special matrix associated the Laplacian matrix must satisfy some additional conditions. For example, in \cite{Zhao2017Event}, $\lambda_{min} (L+D+(L+D)^T)>\frac{2\rho}{\tilde{\alpha}}$, where $L$ denotes Laplacian matrix and $D$ denotes the leader adjacency matrix. To satisfy the condition, the information of Laplacian matrix and leader adjacency matrix has to be known a priori for coupling gains design. 
One may question why not apply a sufficiently small value $\frac{2\rho}{\tilde\alpha}$, without using the global spectra information for solving this problem. It is noticed that a sufficiently small value $\frac{2\rho}{\tilde\alpha}$ means a large value of $\tilde\alpha$, which will directly increase the energy cost of the control. 
Hence it is energy-efficient and of great significance to design a fully distributed approach without using the Laplacian matrix and the leader adjacent matrix. In this paper, we will design an event-triggered communication mechanism to achieve leader-following consensus for second-order MASs and a consensus control protocol with updated coupling gains.

\begin{defi}
	Consensus of a leader-following second-order MAS is said to be asymptotically achieved if both $\lim_{t\rightarrow \infty}\|\hat x_i - \hat{x}_0\| = 0$ and   $\lim_{t\rightarrow \infty}\|\hat v_i - \hat{v}_0\| = 0, i \in \mathbb{N}$ are satisfied for any initial values.  
\end{defi}
\section{Main results}\label{mainresults}
In this section, the main results of this paper are proposed. Generally speaking, the event-triggered transmission strategy consists of two modules \cite{Zhao2017Event}: (a) the consensus control protocol and (b) the event-triggered rule. For {a better understanding}, the overall framework of the proposed event-triggered transmission strategy is illustrated in the Fig.\ref{fig1}, which will be specifically explained in the following subsections.

\begin{figure} 
	\centering
	\includegraphics[%
	width=1\linewidth]{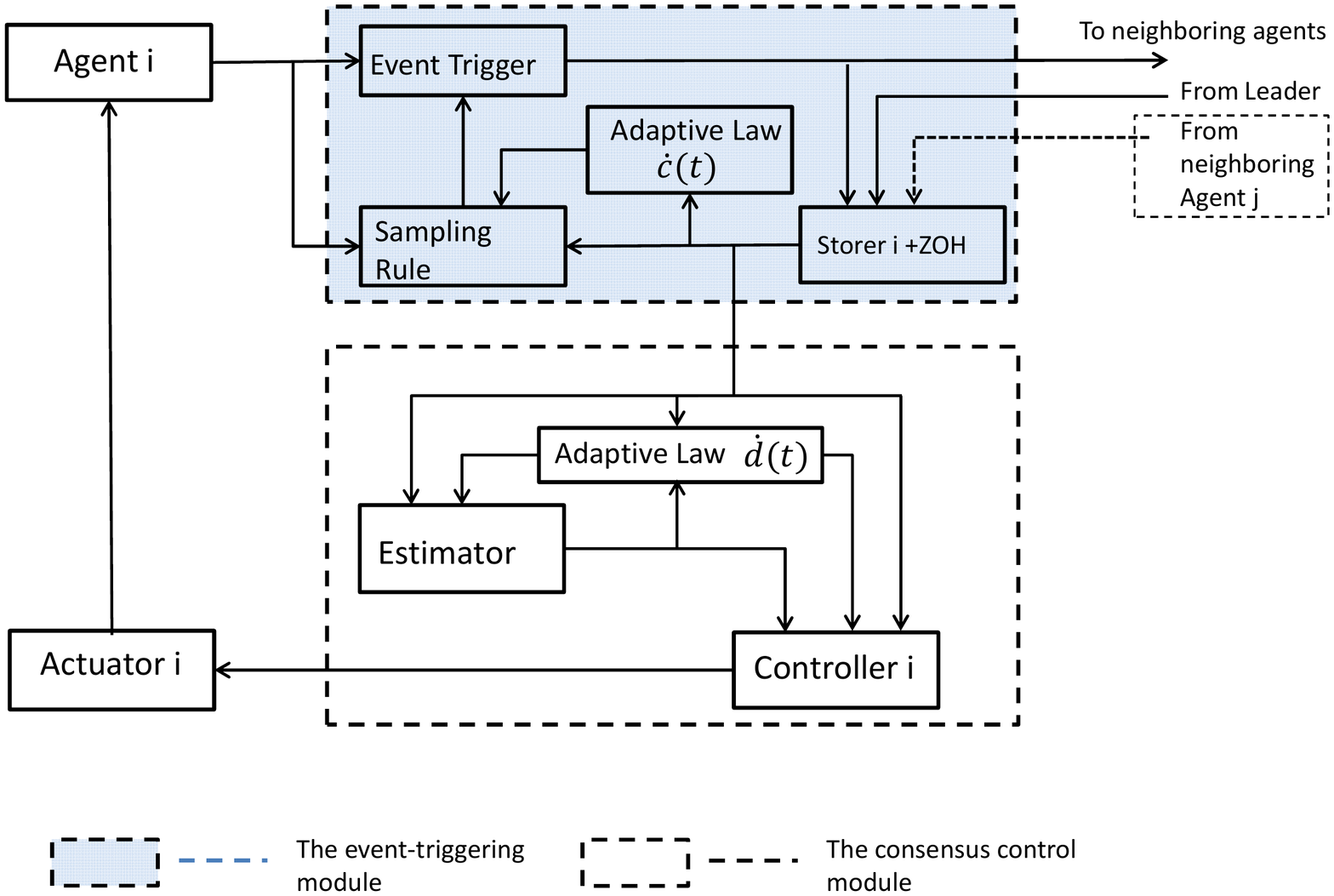}
	\caption{A fully distributed event-triggered transmission strategy for agent-$i$}\label{fig1}
\end{figure}
\subsection{The event-triggered module}
The sampling process of event-trigger mechanisms relies on the event-triggered condition rather than the elapse of a fixed time. Thus the $k$-th sampled-data indicates the data sampled at the $k$-th triggered event.
Denote the $k$-th event-triggered instant of agent-$i$ with $t_k^i$. There exist measurement errors of the event-triggered sampling states $x_i(t_k^i),v_i(t_k^i)$ to its current states $x_i(t),v_i(t)$, which can be defined by $e_{xi}(t) = x_i(t_k^i) -x_i(t)$, and $e_{vi}(t)= v_i(t_k^i) -v_i(t),~~i = 1,2,\ldots,N,$ where $t\in [t_k^i,t_{k+1}^i)$. The next broadcasting instant of agent $i$ is determined by 
\begin{equation}\label{nextins}
t_{k+1}^i = inf\left\lbrace t>t_{k}^i: E_i(t)\geq 0 \right\rbrace,
\end{equation}
where 
\begin{equation}\label{evet-rule}
\begin{split}
&E_i(t)=  \|e_{xi}(t)\|^2 - \underbrace{ d_i(t)sign(d_i)\left\| \sum_{j=1}^{N}L_{ij} x_j(t_k^j) + k_i \left( x_i(t_k^i) - x_0(t) \right)  \right\|^2 }_{\Upsilon(t)} ,
\end{split}
\end{equation} 
and $d_i(t)$ is an updated threshold to be designed and $t_k^j \triangleq \text{arg~min}_p \{t-t_p^j|t\geq t_p^j,p\in \mathbb{N}\}$, i.e., $t_k^j$ is the latest triggering time of agent j before time $t$.
From \eqref{evet-rule}, it can be seen that only relative position information is employed. The workflow of the event-triggered module can be described as follows.
\begin{enumerate}
	\item  The storer $i$ receives the latest state values from the neighboring agents and the leader (if agent $i$ is leader's neighbor). Basing on the information received, storer-$i$ generates the continuous output signals. 
	\item  The adaptive law $\dot d_i(t)$ updates the threshold $d_i(t)$ according to the information from the local storer.
	\item  The sampling rule formulated by \eqref{nextins} processes the sampled-data from the storer with respect to the event-triggered condition.
	\item  The event trigger obtains a triggering signal from the sampling rule and then performs sampling. 
\end{enumerate}
\begin{rmk}
	In the existing literature, there are two forms of control input in the event-triggered control protocol for MAS. One can be formulated by $u_i = \beta\sum_{j\in \mathcal{N}_j}a_{i}\left[ x_j(t_k^i) - x_i(t_k^i)\right] $; another is $u_i = \beta\sum_{j\in \mathcal{N}_j}a_{i}\left[ x_j(t_k^j) - x_i(t_k^i)\right] $, which can be seen that the main difference is the event-triggered sampled time of neighbors' states. 
	In the former scheme, the control input only updates the state signals (from the local agent and the neighboring agents) at the local sampling time instant $t^i_k$; 
	in the latter scheme, these state values need to be updated whenever the local agent samples its state value or receives a new measurement state value from the neighboring agents. The two schemes have their own advantages in different aspects: the latter scheme is superior in the aspect of reducing the burden of networking transmission and the former one serve the purpose of fewer controller updates. Hence, the latter scheme is adopted in this paper from the perspective of alleviating burdens on communication.
\end{rmk}

\begin{rmk}
	In the case that agent $i$ is not the leader's neighbor, the storer $i$ also accounts for zero-order holding of the latest discrete state values received from the neighbors as well as storing them. In the case that it is the leader's neighbor, the store $i$ adds the continuous state values from the leader and the latest discrete state values together and output the sum. It explains why the storer $i$ generates the continuous signals.  
\end{rmk}
\subsection{The consensus control module}
Now we are at the position to present the fully distributed consensus protocol of this paper as follows,
\begin{equation}\label{protocol}
\left\{\begin{aligned}
\dot x_i(t) =& v_i(t),\\
\dot v_i(t) =& f\left(t,x_i(t),v_i(t) \right)  - \alpha c_i(t)\sum_{j=1}^{N}L_{ij}  {x}_j (t_k^j) 
-\alpha c_i(t)k_i  \left[ {x}_i (t_k^i) - {x}_0 (t) \right] - \alpha w_i(t),\\
\dot{w}_i(t)=& -\gamma w_ i -\beta c_i(t)\sum_{j=1}^{N}L_{ij}  {x}_j (t_k^j) 
-\beta c_i(t)k_i  \left[ {x}_i (t_k^i) - {x}_0 (t) \right],
\end{aligned}\right.
\end{equation}
where $\dot w_i$ is the estimator of the networking coupled velocities; $\alpha>0, \beta>0$, $\gamma>0$ are positive coupling gains and $c_i(t)$ is time-varying parameters to be designed.
With Fig.\ref{fig1}, the protocol \eqref{protocol} can be specifically explained by the following workflow. 
\begin{enumerate}
	\item  The adaptive law updates the time-varying gain $c_i(t)$ basing on information from interaction and local estimator;
	\item  The estimator calculates estimates of networking coupling velocities term $w_i(t)$;
	\item  The controller generates the control input and transmits it to the actuator $i$.  
\end{enumerate}

Let  $\tilde{x}_i(t_k^i,t) ={x}_i (t_k^i) - {x}_0 (t)$, $f\left(t,\tilde{x}_i(t),\tilde{v}_i(t) \right) = f\left(t,x_i(t),v_i(t) \right)  - f\left(t,x_0(t),v_0(t) \right)$. The error dynamical equations can be written as 

\begin{equation}\label{protocol-a }
\left\{\begin{aligned}
&\dot {\tilde{x}}_i(t) = v_i(t),\\
&\dot {\tilde{v}}_i(t) = f\left(t,\tilde{x}_i(t),\tilde{v}_i(t) \right)  - \alpha c_i(t)\sum_{j=1}^{N}h_{ij} \tilde{x}_j(t_k^j,t) - \alpha w_i(t),\\
&\dot{w}_i(t)= -\gamma w_ i -\beta  c_i(t)\sum_{j=1}^{N}h_{ij} \tilde{x}_j(t_k^j,t),
\end{aligned}\right.
\end{equation}
where $h_{ij}$ denotes the element of matrix $H$.  
From Lemma \ref{lem - H}, $H$ is positive definite if there is at least one informed agent. Throughout this paper, we make an assumption that there is at least one agent connected to the leader; otherwise, it is impossible to expect the agents in the graph can follow the leader. 
\begin{rmk}
	Since $\sum_{j=1}^{N}L_{ij} = 0$, one can easily derive 
\begin{equation}
\begin{aligned}
\sum_{j=1}^{N}h_{ij} \tilde{x}_j(t_k^i,t) = & \sum_{j=1}^{N}L_{ij}\left[ x_j(t_k^j) - x_0(t)\right]+k_i\left[ x_i(t_k^i) - x_0(t) \right] \\
=& -\sum_{j=1}^{N} w_{ij} \left[ x_j(t_k^j) - x_i(t_k^i)\right]-k_i \left[ x_0(t) - x_i(t_k^i) \right]
\end{aligned}
\end{equation}
\end{rmk}

To facilitate analysis, define a new error state vector $z(t)=[\tilde{x}(t),\tilde{v}(t),{w}(t)]^T \in \mathbb{R}^{3nN}$, where $\tilde{x}(t) = [\tilde{x}_1,\ldots,\tilde{x }_N] \in \mathbb{R}^{nN}$, $\tilde{v}(t) = [\tilde{v}_1,\ldots,\tilde{v}_N]\in \mathbb{R}^{nN}$, ${w}(t) = [{w}_1,\ldots,{w}_N]\in \mathbb{R}^{nN}$. Then the protocol in \eqref{protocol-a } can be recast in the compact form 
\begin{equation}\label{compact-close}
\dot { z}(t) =\widetilde{F}(t,\tilde{x}(t),\tilde{v}(t))+ \widetilde H  z(t) + \widetilde{G} \varepsilon(t),
\end{equation}
where $\varepsilon(t) = [e_{x1} - e_{x0}, \ldots, e_{xN} - e_{x0}]^T\in \mathbb{R}^{nN}$,
$$
\widetilde{H} = \begin{pmatrix}
\textbf{0}&	I_N&				\textbf{0}\\
-\alpha C H &\textbf{0}& -\alpha I_N\\
- \beta CH	&		\textbf{0}		&-\gamma I_N
\end{pmatrix}\in  \mathbb{R}^{3nN\times 3nN}, $$
$\widetilde{G} = \begin{pmatrix}\textbf{0}&-\alpha CH&- \beta C H\end{pmatrix}^T \in \mathbb{R}^{3nN\times nN}$, $\widetilde{F}(t,\tilde{x}(t),\tilde{v}(t))=\begin{pmatrix}\textbf{0}&f\left(t,\tilde{x} (t),\tilde{v}(t) \right)&\textbf{0}\end{pmatrix}^T\in \mathbb{R}^{3nN\times nN}$,
$f\left(t,\tilde{x}(t),\tilde{v}(t) \right)= \left[f\left(t,\tilde{x}_1(t),\tilde{v}_1(t)\right),\ldots,f\left(t,\tilde{x}_N(t),\tilde{v}_N(t) \right)  \right]^T \in \mathbb{R}^{nN} $ and $C$ is diagonal matrix $C=diag\left\lbrace c_1 ,\ldots,c_N \right\rbrace \in \mathbb{R}^{nN}$.

\subsection{Consensus analysis}\label{consensusanalysis}
Based on the event-triggered rule \eqref{evet-rule} and the protocol \eqref{protocol}, the following theorem gives the adaptive laws $\dot c(t)$ and $\dot d(t)$ to guarantee the consensus of the considered MAS in this paper.
\begin{thm}\label{thm1}
	Consider a second-order leader-following multiagent system \eqref{sys1} and \eqref{sys-leader} with the distributed sampling control protocol \eqref{protocol} and the event-triggered sampling rule \eqref{evet-rule}. Suppose that the graph $\mathcal{G}$ is connected and assumption \ref{ASS1}	holds. Then the second-order consensus can be reached under the following distributed adaptive laws:
	
	\begin{eqnarray}\label{dotc}
	\dot c_i(t) =   \tilde{x}_i^T (\beta - \alpha)\sum_{j=1}^{N}h_{ij} \tilde{x}_j(t_k^j,t) + w_i^T(t)\delta \frac{\beta^2 - \alpha^2}{\beta} \sum_{j=1}^{N}h_{ij} \tilde{x}_j(t_k^j,t)  ,
	\end{eqnarray}
	\begin{eqnarray}\label{dotd}	
	\dot d_i(t) = - \xi_i\sum_{j=1}^{N}h_{ij} \tilde{x}_j^T(t^j_k,t)\sum_{j=1}^{N}h_{ij} \tilde{x}_j (t^j_k,t).
	\end{eqnarray}
	where $\delta >0$,  $\xi_i>0$ are constants. 
\end{thm}
\begin{proof}
	Consider the following Lyapunov function candidate 
	\begin{equation}\label{LF}
	V = \frac{1}{2}\tilde z^T(t) \Omega\otimes I_n \tilde{z}(t)  + \sum_{i=1}^{N} \frac{\varpi }{2\zeta_i}(c_i (t) - \hat c_i) ^2 +  \sum_{i=1}^{N} \frac{\omega}{2\xi_i}(d_i (t) + \hat{d}_i)^2,
	\end{equation}
	where $\Omega = \begin{pmatrix}
	\mu& -\varpi &\varpi\\
	* & \eta& -\frac{\alpha}{\beta}\eta \\
	*&*&\eta
	\end{pmatrix}\otimes I_N$, $\varpi$, $\omega$, $\hat c_i$ and $\hat d_i $ are positive constants to be determined.  By letting the parameters in matrix $\Omega$ satisfy $\mu \gg \varpi >0 $, $\eta >0$, it can be guaranteed that $\Omega>0$. The positive semi-definiteness of $V$ in \eqref{LF} can also be easily ensured, which means $V(\tilde z(t) ,\varepsilon, t) \geq 0 $ and  $V(\tilde z(t) ,\varepsilon, t)= 0$ , if and only if $\tilde{z} (t) = 0$ and all $c_i(t) = \hat c_i$ and $d_i(t) = \hat d_i$. 
	For simplicity, we assume $n=1$ in the proof and $I_n$ is equivalent to 1 such that it will be omitted hereafter.
	
	Differentiating \eqref{LF} along the trajectories of \eqref{compact-close} yields 
	\begin{equation}\label{dotV - 1}
	\begin{split}
	\dot V(\tilde z(t) ,\varepsilon, t) =
	& \tilde{z}^T \Omega\dot {\tilde z}(t) + \sum_{i=1}^{N} \frac{\varpi}{\zeta_i}(c_i (t) - \hat c_i) \dot c_i (t)+\sum_{i=1}^{N} \frac{\omega}{\xi_i}(d_i (t) + \hat{d}_i) \dot d_i (t)\\
	=&\tilde{z}^T(t) \Omega {F} + \tilde{z}^T(t) \frac{1}{2}\left(\Omega \widetilde{H}+ \widetilde{H}\Omega^T \right) \tilde{z}(t)+ \tilde{z}^T(t) \Omega \widetilde{G} \varepsilon(t)\\
	&+ \sum_{i=1}^{N} \frac{\varpi }{\zeta_i}(c_i (t) - \hat c_i) \dot c_i (t)+ \sum_{i=1}^{N} \frac{\omega}{\xi_i}(d_i (t) + \hat{d}_i) \dot d_i (t),\\
	\end{split}
	\end{equation}
	where
	\begin{equation*}
	\frac{1}{2}\left(\Omega \widetilde{H}+ \widetilde{H}\Omega^T \right) = \begin{pmatrix}
	\varpi( \alpha- \beta ) C H   & \frac{1}{2}\mu &	\frac{\varpi }{2}(\alpha-\gamma)+\frac{\alpha^2-\beta^2}{2\beta}\eta C H  \\
	*								& -\varpi&-\alpha \eta+\frac{\varpi}{2}+ \frac{\alpha\gamma}{2\beta}\eta\\
	*& *& \frac{\alpha^2 }{\beta}\eta - \gamma\eta
	\end{pmatrix},
	\end{equation*}
	\begin{equation*}
\Omega\widetilde{G}=
\left[ \varpi(\alpha - \beta) C H ,0,\frac{\alpha^2-\beta^2}{\beta}\eta C H \right] ^T,
	\end{equation*}
	 and $\Omega F =\left[ 
	-\varpi f\left(* \right),
	\eta  f\left(* \right),
	-\frac{\alpha}{\beta} \eta  f\left(* \right)
	\right] ^T.$
	From \eqref{dotc}, one obtains    
	\begin{equation}\label{compact-dotc}
	\sum_{i=1}^{N} \frac{\varpi}{\zeta_i}(c_i (t) - \hat c_i) \dot c_i (t) =  \tilde{z}^T\varpi \begin{pmatrix}
	\beta-\alpha \\
	0\\
	\delta\frac{\beta^2- \alpha^2 }{\beta}
	\end{pmatrix}\otimes (C- \hat C) H  (\tilde{x} + \varepsilon)
	\end{equation}
	where $\hat C = diag\{\hat c_1,\ldots,\hat c_N \}$.
	
	Let $\frac{\eta}{\varpi} = \delta $. By substituting \eqref{compact-dotc} into \eqref{dotV - 1} and some simple calculation, one has 
	\begin{equation}\label{dotV - 2}
	\begin{split}
	\dot V(\tilde z(t) ,\varepsilon, t) 
=	
	&\tilde{z}^T(t) \Omega F +\tilde{z}^T(t) \frac{1}{2}\left(\Omega \bar{H}+ \bar{H}\Omega^T \right) \tilde{z}(t)+ \tilde{z}^T(t) \Omega \bar{G} \varepsilon(t)\\
	& + \sum_{i=1}^{N} \frac{\omega}{\xi_i}(d_i (t) + \hat d_i) \dot d_i (t), 
	\end{split}
	\end{equation}
	where 
	\begin{equation*}
	\frac{1}{2}\left(\Omega \bar{H}+ \bar{H}\Omega^T \right) = \begin{pmatrix}
	\varpi( \alpha- \beta ) \hat  C H   & \frac{1}{2}\mu &	\frac{\varpi }{2}(\alpha-\gamma)+\frac{\alpha^2-\beta^2}{2\beta}\eta \hat C H  \\
	*								& -\varpi&-\alpha \eta+\frac{\varpi}{2}+ \frac{\alpha\gamma}{2\beta}\eta\\
	*& *& \frac{\alpha^2 }{\beta}\eta - \gamma\eta
	\end{pmatrix},
	\end{equation*}
	$\Omega\bar{G}=
	\left[ \varpi(\alpha - \beta) \hat C H ,0,\frac{\alpha^2-\beta^2}{\beta}\eta \hat C H \right] ^T$. 
	Following Assumption \ref{ASS1}, one gets
	\begin{equation}\label{OF}
	\begin{split}
	\tilde{z}^T(t) \Omega F = 		&\sum_{i=1}^{N}\left( -\varpi \tilde{x}_i^T + \eta \tilde{v}_i^T - \frac{\alpha}{\beta} \omega_i^T  \right) \left[ f\left(t,x_i(t),v_i(t) \right)  - f\left(t,x_0(t),v_0(t) \right)\right] \\
	\leq	& \varsigma\left( -\varpi \tilde{x}_i^T + \eta \tilde{v}_i^T -\frac{\alpha}{\beta}\omega_i^T\right)\tilde{v}_i  +\rho \left(\varpi \|\tilde{x}_i\|^2+ \eta\|\tilde{v}_i\tilde{x}_i\|+\frac{\alpha}{\beta}\|\omega_i\tilde{x }_i\|\right)  \\
	\leq	&\varsigma\left( -\varpi \tilde{x}_i^T + \eta \tilde{v}_i^T -\frac{\alpha}{\beta}\omega_i^T\right)\tilde{v}_i +\kappa\|z\|^2,
	\end{split}
	\end{equation}
	where $\kappa = max\{\rho (\varpi+\frac{\eta}{2}+\frac{\alpha}{2\beta}), \frac{\eta}{2}, \frac{\alpha}{2\beta}\}$.

	Recasting the event-triggered condition \eqref{evet-rule} in the compact form, one obtains
	\begin{equation} \label{dotV - a1} 
-\omega\varepsilon^T(t)\varepsilon(t)+\omega\bar D\tilde{x}^T(t_k,t)H^2 \tilde{x}(t_k,t)\geq 0. 
	\end{equation} 
	where $\bar {D} = diag\{d_1sgn(d_1),\ldots,d_N sgn(d_N)\}$.
	
	Besides, substituting the adaptive law \eqref{dotd} into  $\sum_{i=1}^{N} \frac{\omega}{\xi_i}(d_i (t) + \hat{d}_i) \dot d_i (t)$, one gets
	\begin{equation}\label{dotV - d}
	\begin{split}
	\sum_{i=1}^{N} \frac{\omega}{\xi_i}(d_i (t) + \hat{d}_i) \dot d_i (t) =& - \omega\left[\tilde{x} (t) +  \varepsilon(t) \right]^T(D+\hat D)  H^2 \left[\tilde{x} (t) + \varepsilon(t) \right],
	\end{split}
	\end{equation}
	where $D=diag\{d_1,\ldots,d_N\},\hat{D}=diag\{\hat d_1,\ldots,\hat d_N\}$ Combining \eqref{dotV - a1} and \eqref{dotV - d}, then
	
	\begin{equation} \label{dotV - a2} 
	\begin{split}
	\sum_{i=1}^{N} \frac{\omega}{\xi_i}(d_i (t) + \hat{d}_i) \dot d_i (t) 
	\leq&	- \omega \varepsilon^T(t) \varepsilon(t) -\omega\left[\tilde{x} (t) + \varepsilon(t) \right]^T\underbrace{(D-\bar{D}+\hat D)}_{\Delta_D}  H^2 \left[\tilde{x} (t) + \varepsilon(t) \right],
	\end{split}
	\end{equation} 
	holds. From the definitions of $D$ and $\bar D$, $\Delta_D=\hat D$, if $ d_i(t)\geq 0$; $\Delta_D = \hat D-2\bar D$, otherwise. 
	Besides, recalling the definition $\eqref{dotd}$, one can observe that $\dot{d}_i(t)\leq 0$, which means the value of $d_i(t)$ will never increase. Then $d_i(t)\leq d_i(t_0)$, where $d_i(t_0)$ denotes the value of $d_i(t)$ at initial time instant $t_0$. Here, by choosing the constants $\hat{d}_i\geq 2 d_i(t_0)sgn(d_i(t_0))\geq 2d_isgn(d_i)$, it is not hard to derive $\Delta_D\geq 0, \forall d_i(t)\in (-\infty,+\infty)$. 
	Substituting \eqref{dotV - a2} into \eqref{dotV - 2}, one obtains
	\begin{equation}\label{5-a}
	\begin{split}
	\dot V(\tilde z(t) ,\varepsilon, t) 
	=&	\tilde{z}^T(t) \Omega F + \tilde{z}^T(t) \frac{1}{2}\left(\Omega \bar{H}+ \bar{H}\Omega^T \right) \tilde{z}(t)+\tilde{z}^T(t) \Omega \bar{G} \varepsilon(t)- \omega \varepsilon^T(t)  \varepsilon(t)  \\
	&-\omega\left[\tilde{x} (t) + \varepsilon(t) \right]^T\Delta_D  H^2 \left[\tilde{x} (t)  + \varepsilon(t) \right]\\ 
	\leq &\begin{pmatrix}
	z^T(t)&\epsilon^T
	\end{pmatrix}\Pi\begin{pmatrix}
	z(t)\\\epsilon
	\end{pmatrix}   ,
	\end{split}
	\end{equation}
	
	where 
	$$
	\Pi = \begin{pmatrix}
	\Pi_{11}&\Pi_{12}\\
	*& \Pi_{22}
	
	\end{pmatrix}
	$$
	
	\begin{equation*}
	\Pi_{11}= \begin{pmatrix}
	\varpi(\alpha- \beta ) \hat{C} H+\kappa I_N -\omega\Delta_D H^2  &(\frac{1}{2}\mu -\frac{\varpi \varsigma}{2})I_N &	\frac{\varpi }{2}(\alpha-\gamma)+\frac{\alpha^2-\beta^2}{2\beta}\eta \hat C H  \\
	*								& \left( -\varpi + \varsigma \eta + \kappa \right)  I_N &-\alpha \eta+\frac{\varpi}{2}+ \frac{\alpha\gamma}{2\beta}\eta - \frac{\varsigma\alpha}{2\beta}\\
	*& *& \left( \frac{\alpha^2 }{\beta}\eta - \gamma\eta+\kappa\right) I_N
	\end{pmatrix},  
	\end{equation*}
	\begin{equation*}
	\Pi_{12} = \begin{pmatrix}
	\varpi(\alpha - \beta) \hat{C} H  -\omega\Delta_D H^2  \\
	0\\
	\frac{\alpha^2-\beta^2}{\beta}\eta \hat{C} H
	\end{pmatrix},~~~\Pi_{22}= - \omega(I_N+ \Delta_D H^2).
	\end{equation*}
	
	By properly selecting the parameters $\mu,\varpi$, $\eta$, $\omega$, $\hat{c}_i$, $\hat{d}_i$ in Lyapunov function candidate \eqref{LF}, not hard to derive that $\Pi<0$ holds with the help of Lemma \ref{LEM - L} and Lemma \ref{lem - H}.  It is observed that $V(\tilde z(t) ,\varepsilon, t) $ and $\dot V(\tilde z(t) ,\varepsilon, t) $ satisfy conditions (1) and (2) of Lemma.\ref{lem-BABA}, respectively. To verify the condition (3) in Lemma \ref{lem-BABA}, the following analysis is needed. From \eqref{5-a} and $\Pi<0$, one may easily derive that $\tilde{z}, c(t),d(t)$ are all bounded. 
	Also, from \eqref{dotV - a2}, the boundedness of $\tilde{z}(t)$ is used to derive that $\epsilon(t)$ is bounded.
	Then from \eqref{compact-close}, \eqref{dotc} and \eqref{dotd}, one can derive the boundedness of $\dot {\tilde{z}} (t),\dot c(t),\dot d(t)$. 
	Then, by invoking \eqref{dotV - 1}, it is finally obtained that $\ddot V(\tilde z(t) ,\varepsilon, t)$ is bounded, i.e., condition (3) in lemma \ref{lem-BABA} is satisfied. Hence the proof can be completed. 
	\begin{rmk}
		One may question that the matrix $H$ including the Laplacian matrix $L$ as well as the matrix $D$ from the whole graph information and topology needs to be known by each agent when solving LMIs to guarantee $\Pi <0$ and the method could not considered as a fully distributed method. It should be pointed out that the parameters obtained by solving $\Pi<0$ are based on the fact $H>0$. Namely, as long as the matrix $H$ is positive definite, the proposed method guarantees the consensus of the network. It is well known $H$ is positive definite if there is at least one informed agent, which is assumed throughout the paper. Therefore, the method is fully distributed.
	\end{rmk}
	\begin{rmk}
		From the Themrem \ref{thm1}, it can be seen that the event-triggered second-order consensus in the considered leader-following MAS can be reached under the distributed adaptive laws \eqref{dotc} and \eqref{dotd} without requiring any centralized conditions, like some existing literature \cite{Zhao2017Event}, \cite{Li2015Event}, \cite{Guo2014A} and \cite{Zhu2014Event}. In the whole networking control design, including the event-triggered rule and the consensus protocol, only local information of neighboring agents is used.  
	\end{rmk}
	\begin{rmk}
		One may notice that the dimension of $\Pi$ is $4N$, which may result in that the selection of the parameters in the Lyapunov function candidate is not easy. As a matter of fact, the selecting of parameters can transfer to the problem of solving feasible solutions of multiple linear matrix inequations. By solving these LMIs, one can easily obtain proper feasible solutions. Also, we provide an example of the feasible solutions for these parameters in the numerical result section. 
	\end{rmk}
	
\end{proof}

The following theorem shows the existence of a lower bound of inter-event times, which means that the Zeno behavior is excluded in Theorem \ref{thm1}.
\begin{thm}\label{thm2}
	With the event-triggered consensus protocol and the conditions given in Theorem \ref{thm1},  there exists no agent in MAS \eqref{sys1} that exhibits Zeno behavior during the consensus process. That is, for each agent $i\in \mathcal{V}$, the inter-event time $\Delta_{k}^i = t_{k+1}^i - t_{k}^i = \tau>0 . $  
\end{thm}
\begin{proof}
	Suppose the velocities of all agents in the network considered are bounded by $M_v>0$. 
	At the triggering time instants $\{t^i_k\}_{k=0}^{\infty}$, $e_{xi}(t^i_k) = 0, i=1,2,\ldots,N$ from the definition of $e_{xi}(t)$. 
	In each interval $t\in \left[t_k^i,t_{k+1}^i \right)$, one gets
	\begin{equation}\label{4-2-1}
	\begin{split}
	\|e_i(t)\| =& \|\int_{t_{k}^i}^{t}\dot e_i(s)ds \|\leq \int_{t_{k}^i}^{t}\|\dot e_i(s)\|ds \leq  \int_{t_{k}^i}^{t}\|\dot x_i(s)\|ds\\
	\leq&    \int_{t_{k}^i}^{t}\|v_i(s)\|ds \leq M_v(t-t_{k}^i), \forall t\in \left[t_k^i,t_{k+1}^i \right)
	\end{split}
	\end{equation}
	According to event triggered rule \eqref{nextins}, the next event will not be triggered until trigger function $E_i(t)= 0$, which means that, for agent i the next sampling time instant $t= t_{k+1}^i$ is at the moment when $\|e_{xi}(t_{k+1}^i,t)\|^2 =  {\varUpsilon_i(t)}$ holds, where $\Upsilon(t)$ is defined in \eqref{evet-rule}.
	
	Assume that before consensus is reached, there exists a positive constant $\underline{\varUpsilon}_i$ such that $\varUpsilon_i (t)\geq \underline{\varUpsilon}_i>0 $, for some $t\in \{t_l^i \}_{l=0}^{\infty}$; otherwise, $\varUpsilon_i (t_{k'}^i) =0$ for some $t\in \{t_{k'}^i \}_{k'=0}^{\infty}$. At the event time $t_{k'}^i$, the consensus has been achieved and there is no need to trigger the event.  That is to say, before the consensus being achieved, from \eqref{4-2-1}, one gets
	\begin{equation}\label{4-2-2}
	\begin{split}
	\left[M_v(t_{k+1}^i - t_k^i) \right]^2& \geq \|e_{xi}(t_{k+1}^i)\|^2 =  {\varUpsilon_i(t)}\geq{\underline{\varUpsilon}_i}>0.
	\end{split}
	\end{equation}
	Now we will prove $\lim_{k\rightarrow \infty}t_k^i = \infty$ by contradiction. Assuming $\lim_{k\rightarrow \infty}t_k^i =T^i< \infty$, one can easily drive
	$
	{\underline{\varUpsilon}_i} \leq   \left[M_v(t_{k+1}^i - t_k^i) \right]^2
	$, which implies ${\underline{\varUpsilon}_i}\leq0$. This contradicts \eqref{4-2-2}. Consequently, $\lim_{k\rightarrow \infty}t_k^i = \infty$ is proven. 
	Assuming $\Delta_{k}^i \rightarrow 0$ and invoking \eqref{4-2-2}, one can verify that $ \underline{\varUpsilon}_i \leq0$, which contradicts the condition $\underline{\varUpsilon}_i>0$. $\Delta_{k}^i = t_{k+1}^i - t_{k}^i = \tau>0$ therefore holds.
	This completes the proof of Theorem \ref{thm2}. 
\end{proof}


\section{Numerical results}\label{simu}

In this section, a numerical example is presented to illustrate the feasibility and effectiveness of the proposed mechanism. 
We consider a multi-agent system with 1 leader and 6 agents. To verify that it is fully distributed without requiring the spectra of Laplacian matrices, we use the following two graphs whose eigenmatrices are $\mathcal{G}_1 $  and $\mathcal{G}_2 $, which are given by
$$\mathcal{G}_1  = \begin{pmatrix}6&-1&-2&-1&-2&0\\-1&8&-3&0&0&-4\\-2&-3&5&0&0&0\\-1&0&0&4&-3&0\\-2&0&0&-3&6&-1\\0&-4&0&0&-1&5 \end{pmatrix},
\mathcal{G}_2 =  \begin{pmatrix}2&-1&-1&0 &0&0\\-1&5&-3&-1&0&0\\-1&-3&6&-1&-1&0\\0&-1&-1&3&0&-1\\0&0&-1&0&5&-3\\0&0&0&-1&-3&4\end{pmatrix}.$$
Accordingly, leader weight matrices are set to be $diag\{2,0,0,0,0\}$ and $diag\{0,0,5$ ,$0,0\}$.
The nonlinear dynamics of agents is the pendulum model which is given by $f(t,x_i,v_i) = -\frac{g}{l}sin(x_i) - \frac{k}{m} v_i $, where $g,k,l,m$ are the gravitational acceleration, the coefficient parameters, the length and the mass of the rob, respectively. It is easy to verify that such a nonlinear dynamic model satisfies assumption \ref{ASS1}. Here, we take $g=9.8, k=0.1, l= 4, m=1$ here. To find a group of feasible parameters satisfying $\Pi<0$ in Theorem.\ref{thm1}, one can use LMI toolbox in MATLAB. Here, we present a group of the parameters for the Lyapunov function candidate \eqref{LF}: $\mu = 2, \varpi = 1.5, \eta  = 20.5, \omega = 40, \Delta_D = 0.5I_6,\varsigma = -0.5, \rho = -2$ and $\hat C = \frac{1}{20} H^{-1}$. For the parameters in consensus control protocol \eqref{protocol} and the adaptive law of the thresholds \eqref{dotd}, in the simulation, their values are taken as $\alpha = 1, \beta = 30, \gamma = 35, \delta = 13.67, \xi = \left[\xi_1,\ldots,\xi_6 \right]  = \left[ 0.5,0.2,0.4,0.3,0.5,0.6\right] $. Besides, the initial positions and velocities of the leader and followers are randomly generated between $\left[ -1,1 \right]$. 
From Fig.\ref{Positions}, it can be observed that all the follower agents can track the position of the leader under both graph $\mathcal{G}_1$ and graph $\mathcal{G}_2$ without retuning the parameters. Also, Fig. \ref{Velocities} shows that the tracking performance of velocities is also guaranteed. 
In Fig.\ref{3D}, the tracking errors of positions and velocities of 6 agents with graph $\mathcal{G}_1$ are presented and it demonstrates the second-order consensus performance of the proposed method. 

Under $\mathcal{G}_1$, the states of adaptive protocol coupling gains are presented in Fig.\ref{C-Gains}, where the distributed control gains $c_i$ adaptively converge to proper ones. Fig.\ref{Positions} and Fig.\ref{Velocities} demonstrates that the second-order leader-following consensus can be achieved with the proposed network protocol in this paper.
\begin{figure*}[htbp]
	\centering
	\subfloat[Positions of the leader and followers under $\mathcal{G}_1$]{\includegraphics[width = 0.5\linewidth]{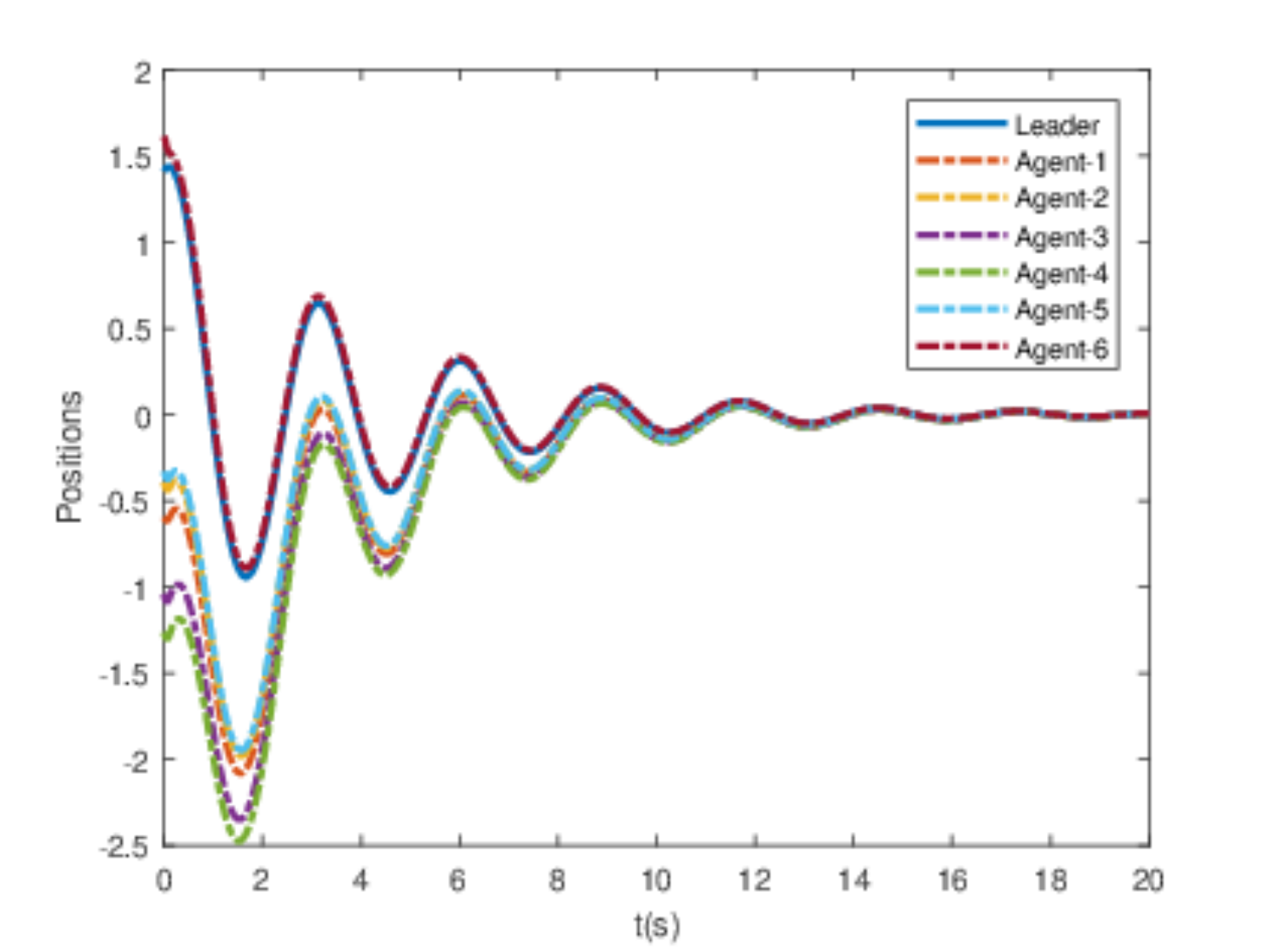}}
	\centering
	\subfloat[Positions of the leader and followers under $\mathcal{G}_2$]{\includegraphics[width = 0.5\linewidth]{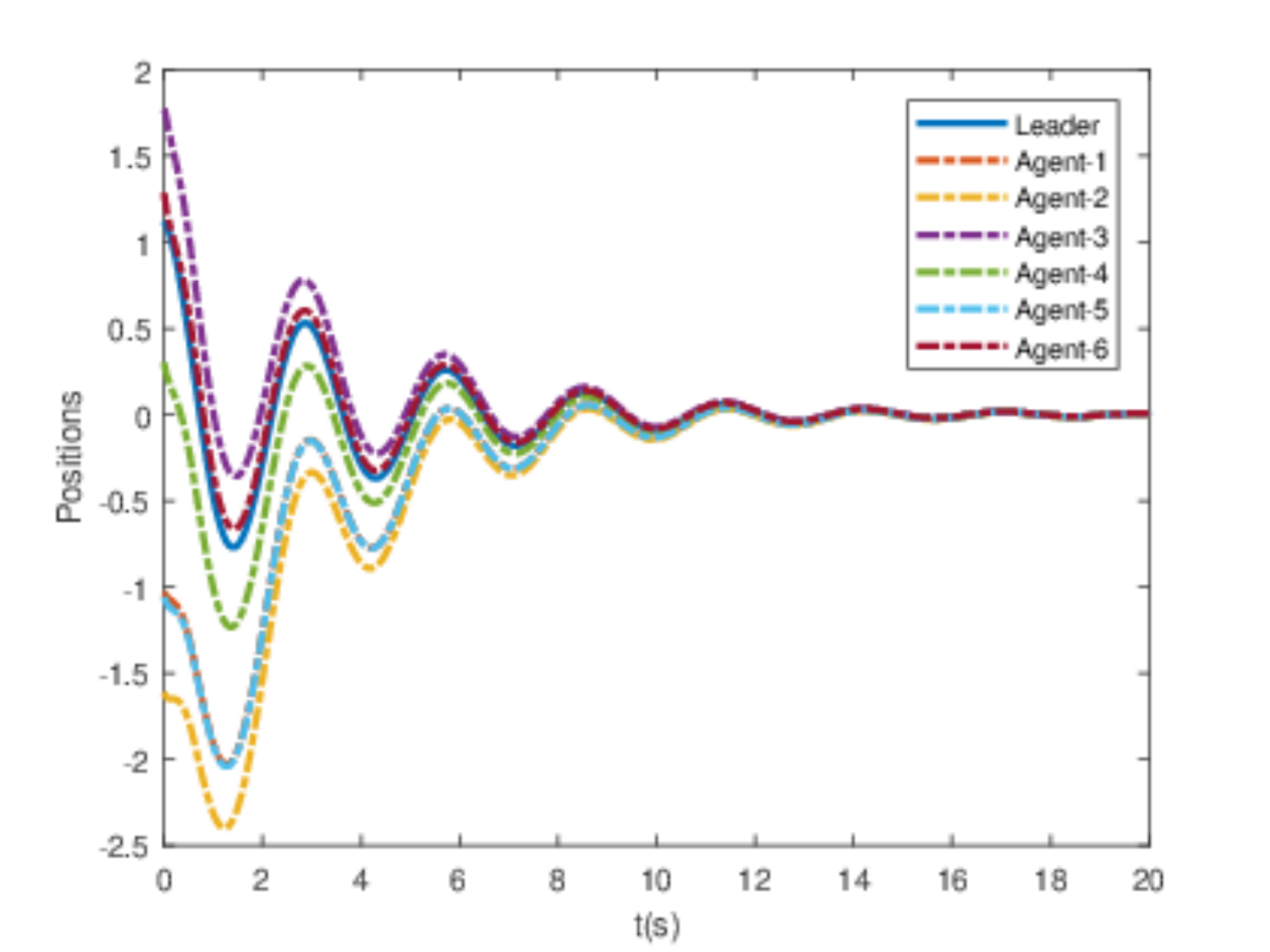}}
	\caption{Consensus of positions under different topologies}\label{Positions}
\end{figure*}
\begin{figure*}[htbp]
	\centering
	\subfloat[Velocites of the leader and followers under $\mathcal{G}_1$]{\includegraphics[width = 0.5\linewidth]{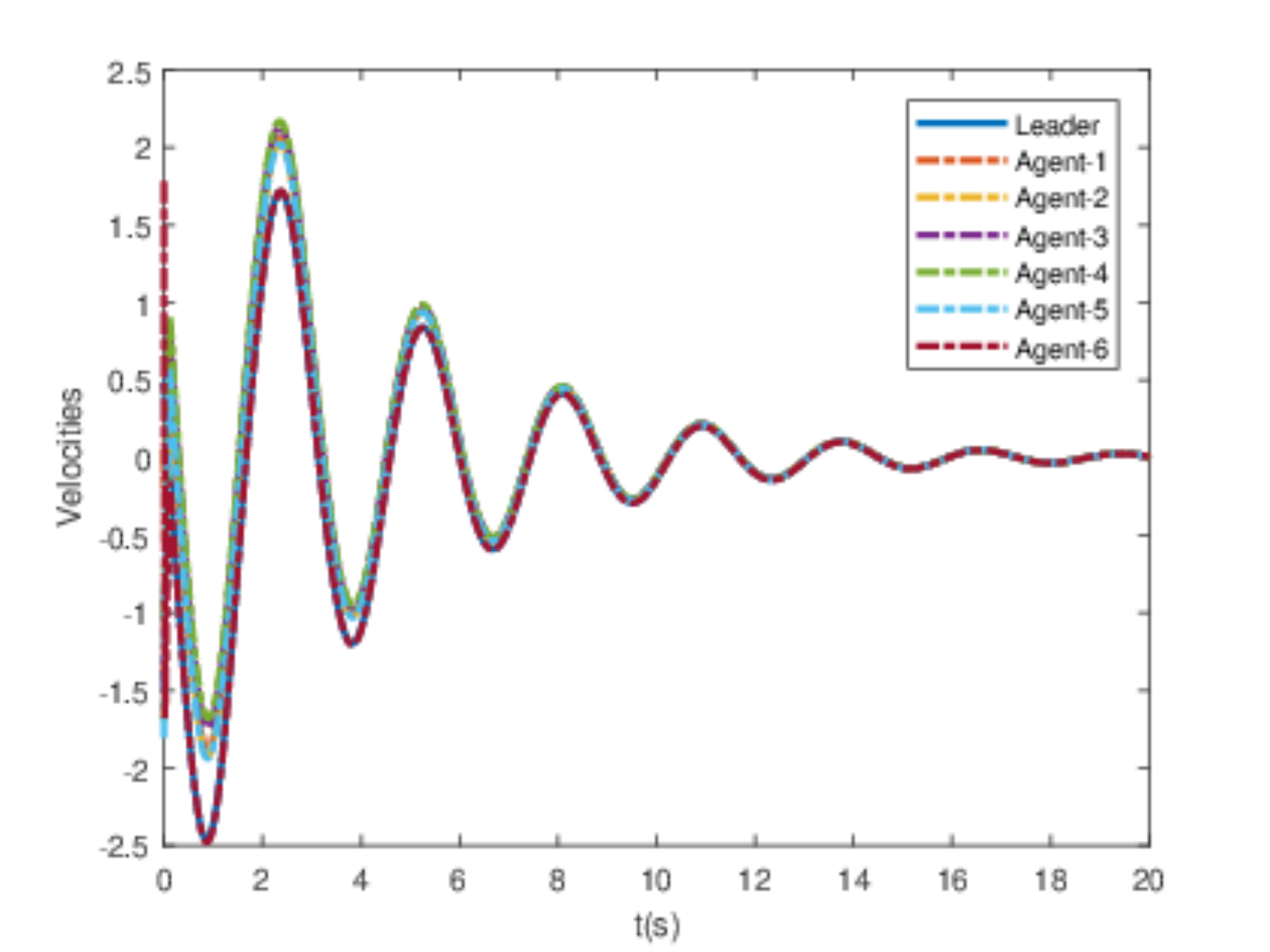}}
	\centering
	\subfloat[Velocites of the leader and followers under $\mathcal{G}_2$]{\includegraphics[width = 0.5\linewidth]{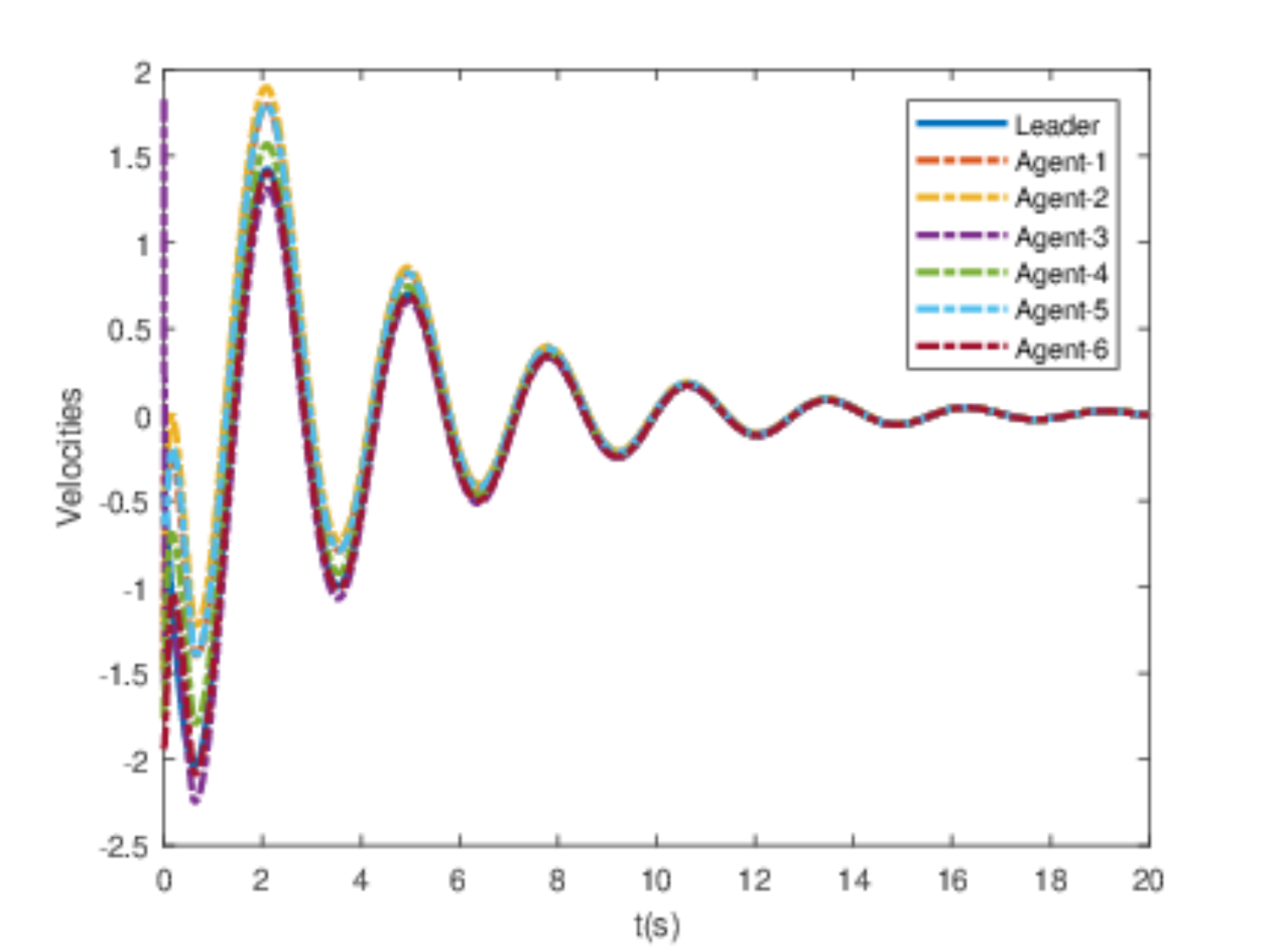}}
	\caption{Consensus of velocities under different topologies}\label{Velocities}
\end{figure*}
\begin{figure}[!ht]
	\hfil
	\subfloat{\includegraphics[width=3in]{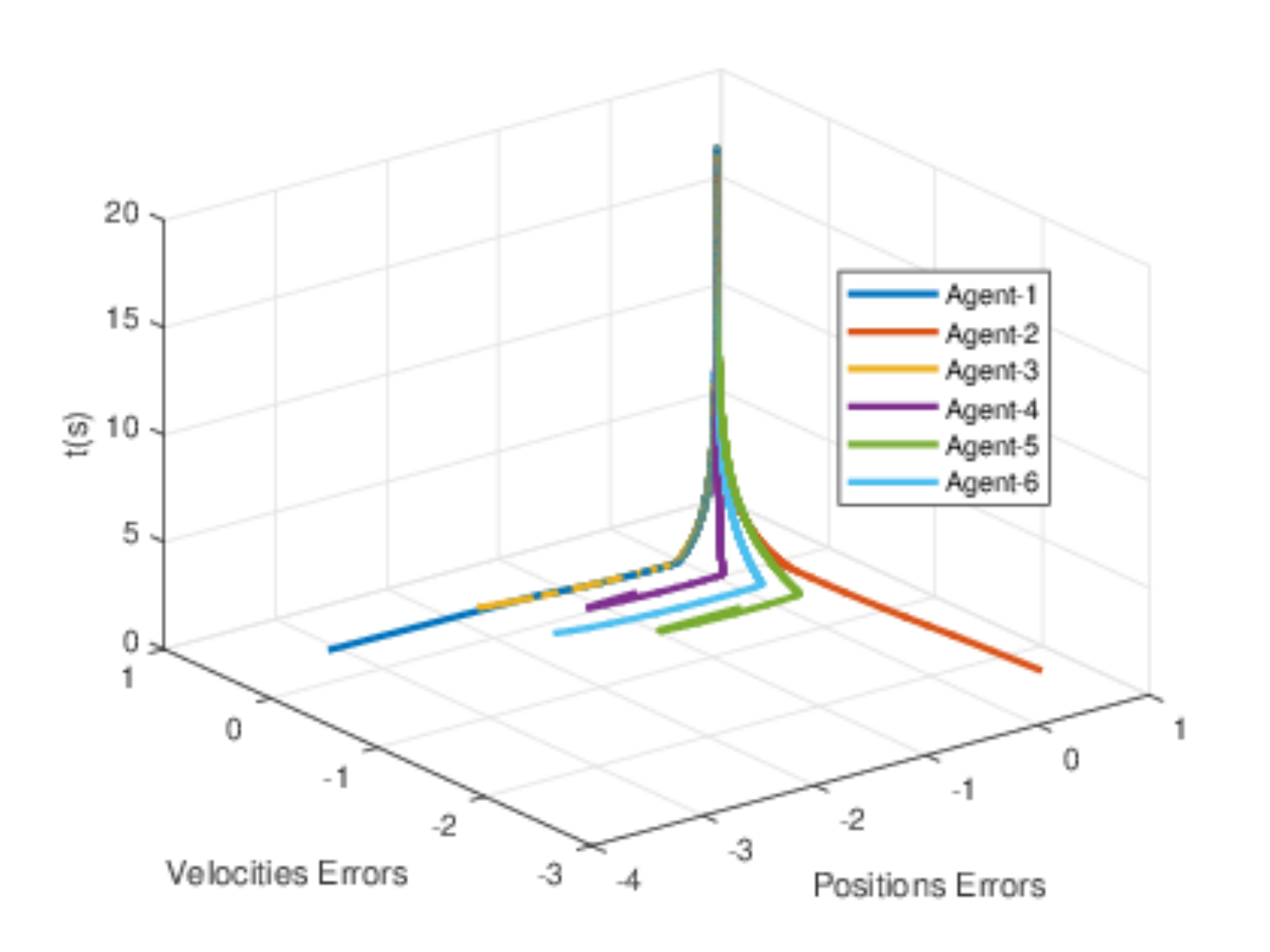}}%
	\caption{The second-order consensus under the proposed control protocol}
	\label{3D}
\end{figure}
\begin{figure}[!ht]
	\hfil
	\subfloat{\includegraphics[width=4in]{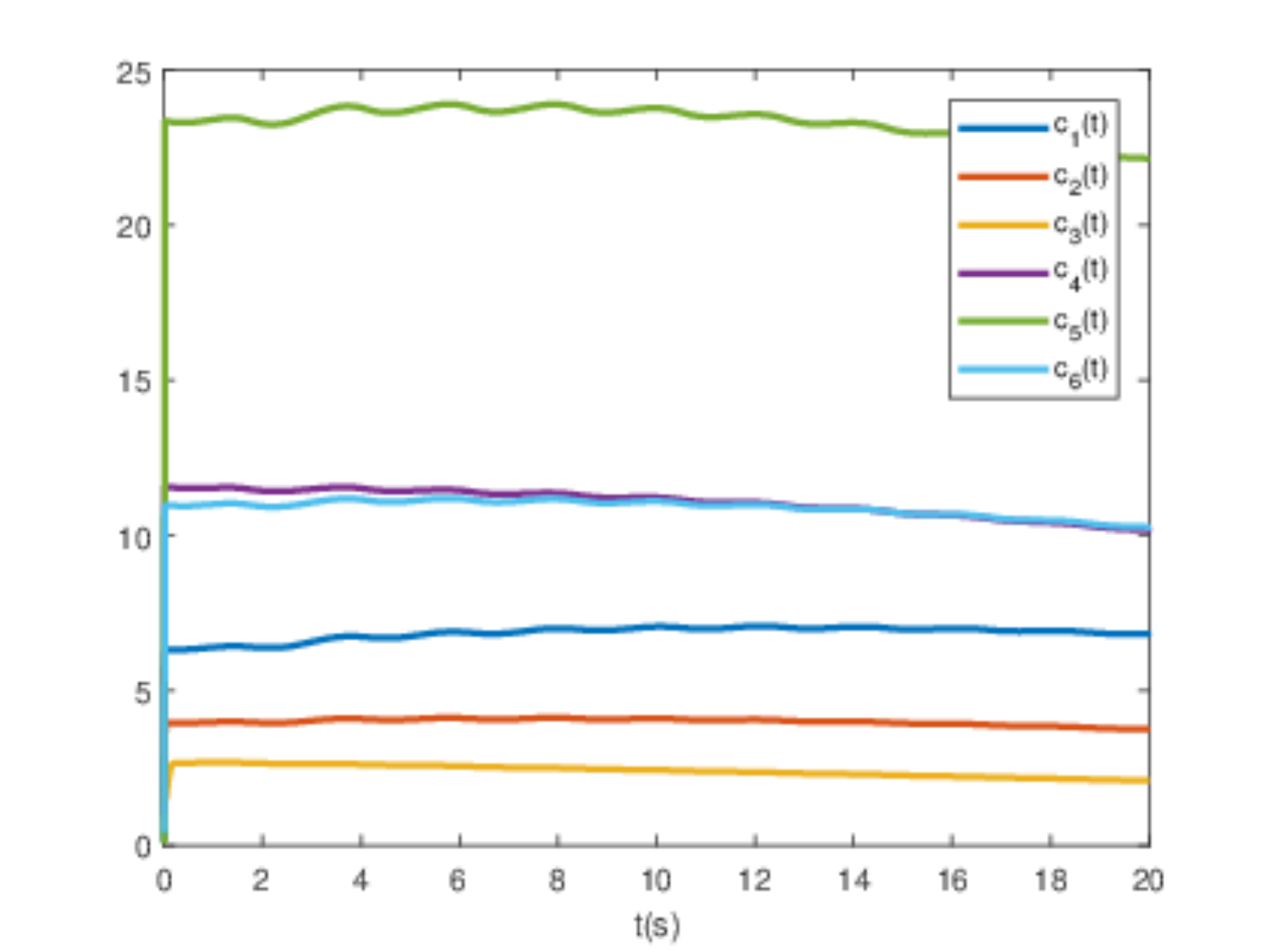}}%
	\caption{Consensus protocol coupling gains $c_i(t)$ under $\mathcal{G}_1$ }
	\label{C-Gains}
\end{figure}
\begin{figure}[!ht]
	\hfil
	\subfloat{\includegraphics[width=4in]{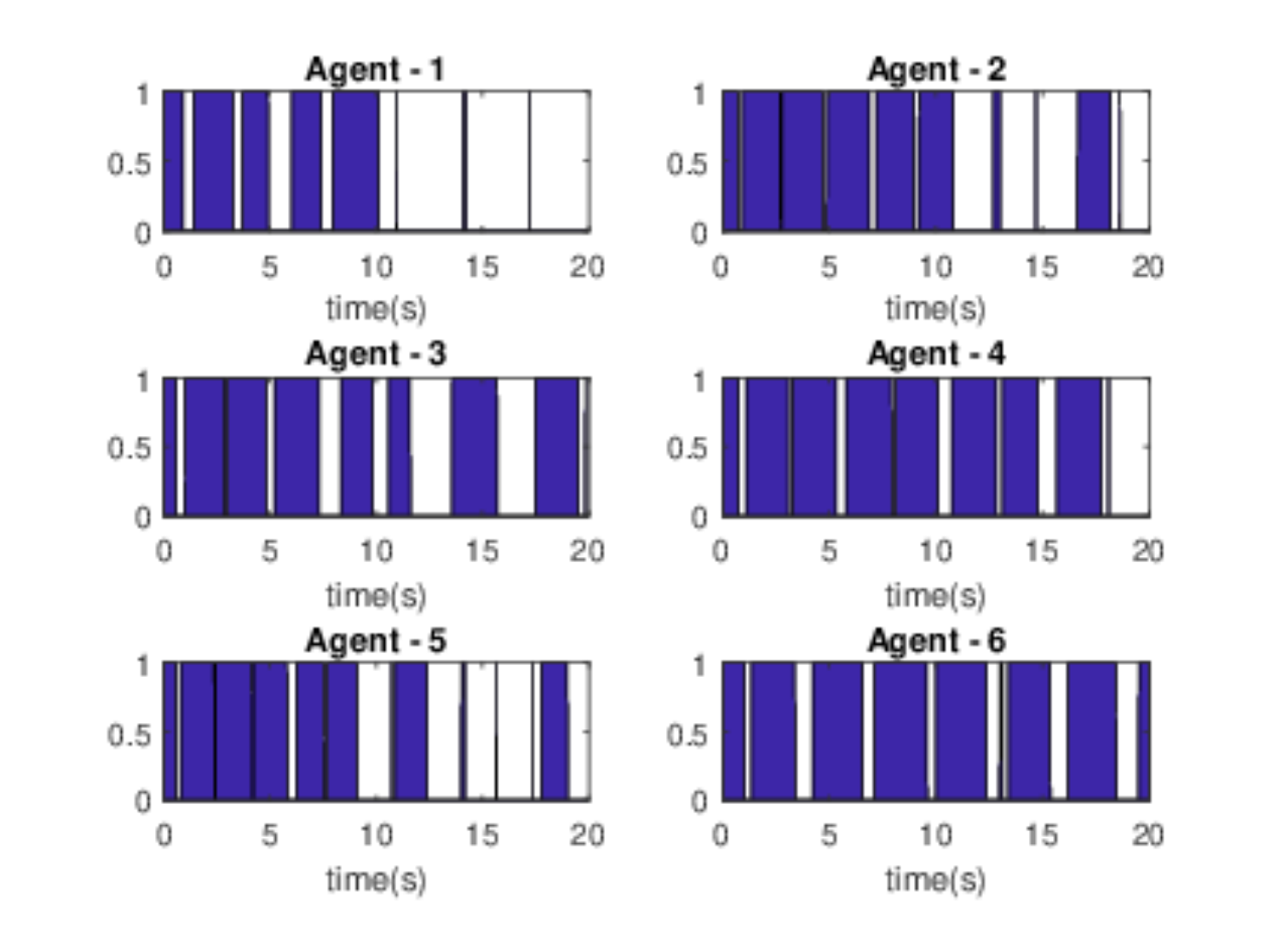}}%
	\caption{The state of triggered events of broadcasting signals under $\mathcal{G}_1$}
	\label{Comm}
\end{figure}

To show the effectiveness of ETM on reducing the frequency of interagent exchanges, Fig.\ref{Comm} presents the states of the events that each agent broadcasts its state to others under topology graph $\mathcal{G}_1$, where the blue areas represent that the predefined events are triggered.

For comparison, we also conduct the simulation for the ETM with the constant event-triggered thresholds in \cite{Zhao2017Event} (see Eq.(7)). In this work, the number of broadcasting interacting signals of each follower is negative related to the event-triggered thresholds parameters. By taking agent-$1$ as an example, the relationship between $\varrho$ and the number of the triggered events $R$ is presented in Fig.\ref{VAR-R} under $\mathcal{G}_1$ and $\mathcal{G}_2$, respectively. Note that, to facilitate analysis, we use a new defined parameter $\varrho\in \left[ 0,1\right] $ to replace  $\varrho_1,\varrho_2$ where $ \varrho_1 = 0.12 \varrho, \varrho_2=0.18 \varrho$ in this simulation. 
Also, the adaptive threshold $d_1(t)$ of agent-$1$ under $\mathcal{G}_1$ and $\mathcal{G}_2$ are accordingly given in Fig.\ref{thre-gains-in2}. 
The comparison of ETM proposed in this paper and its counterpart in \cite{Zhao2017Event} demonstrates that the adaptive triggering thresholds are free of using the spectra of Laplacian matrices, which verifies the effectiveness of the proposed control protocol.  

\begin{figure}[!ht]
	\hfil
	\subfloat{\includegraphics[width=3in]{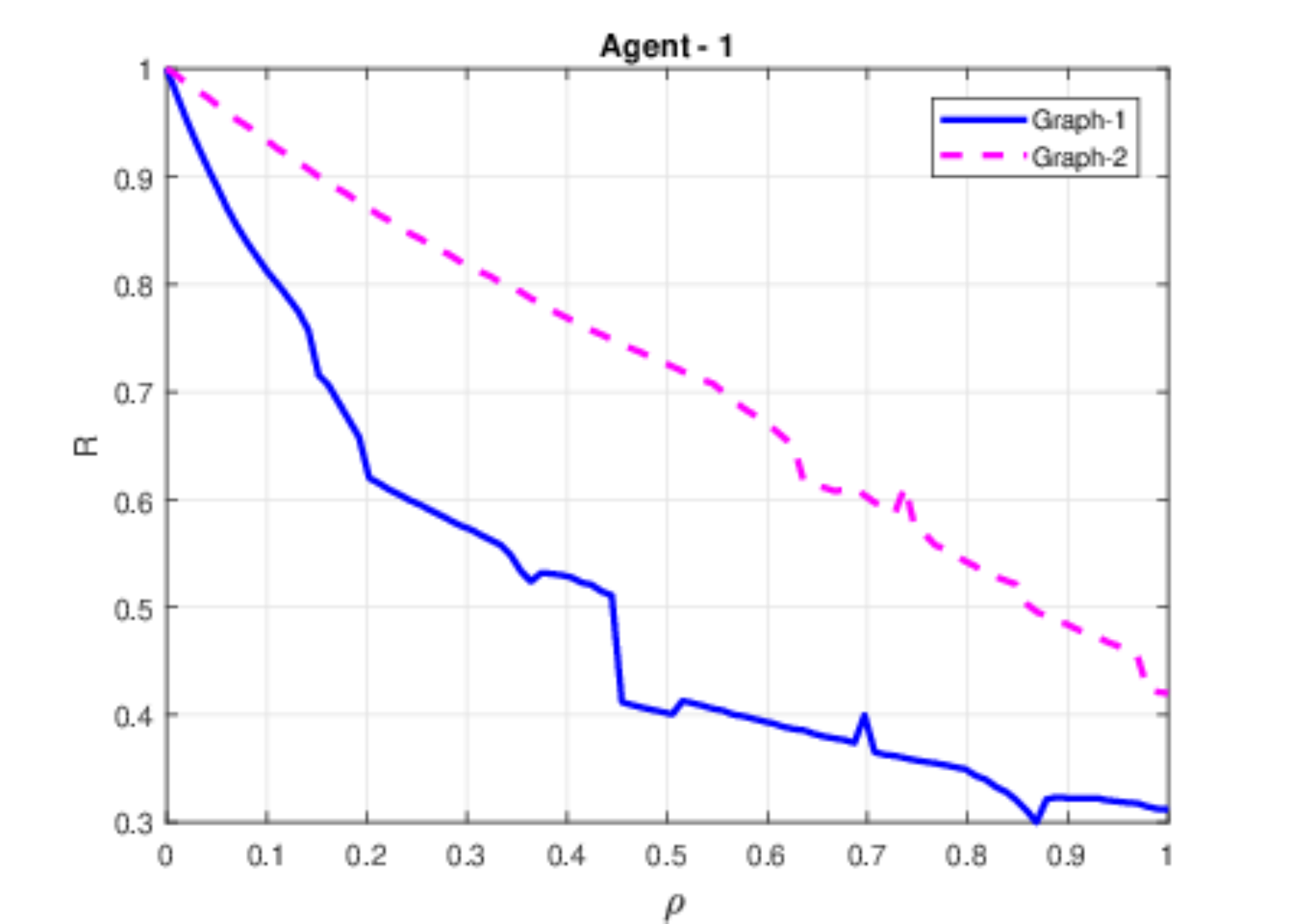}}%
	\caption{Ratio of broadcasting interacting signal and event-triggered criterion $\rho$ of the ETM in \cite{Zhao2017Event}}
	\label{VAR-R}
\end{figure}
\begin{figure}[!ht]
	\hfil
	\subfloat{\includegraphics[width=3in]{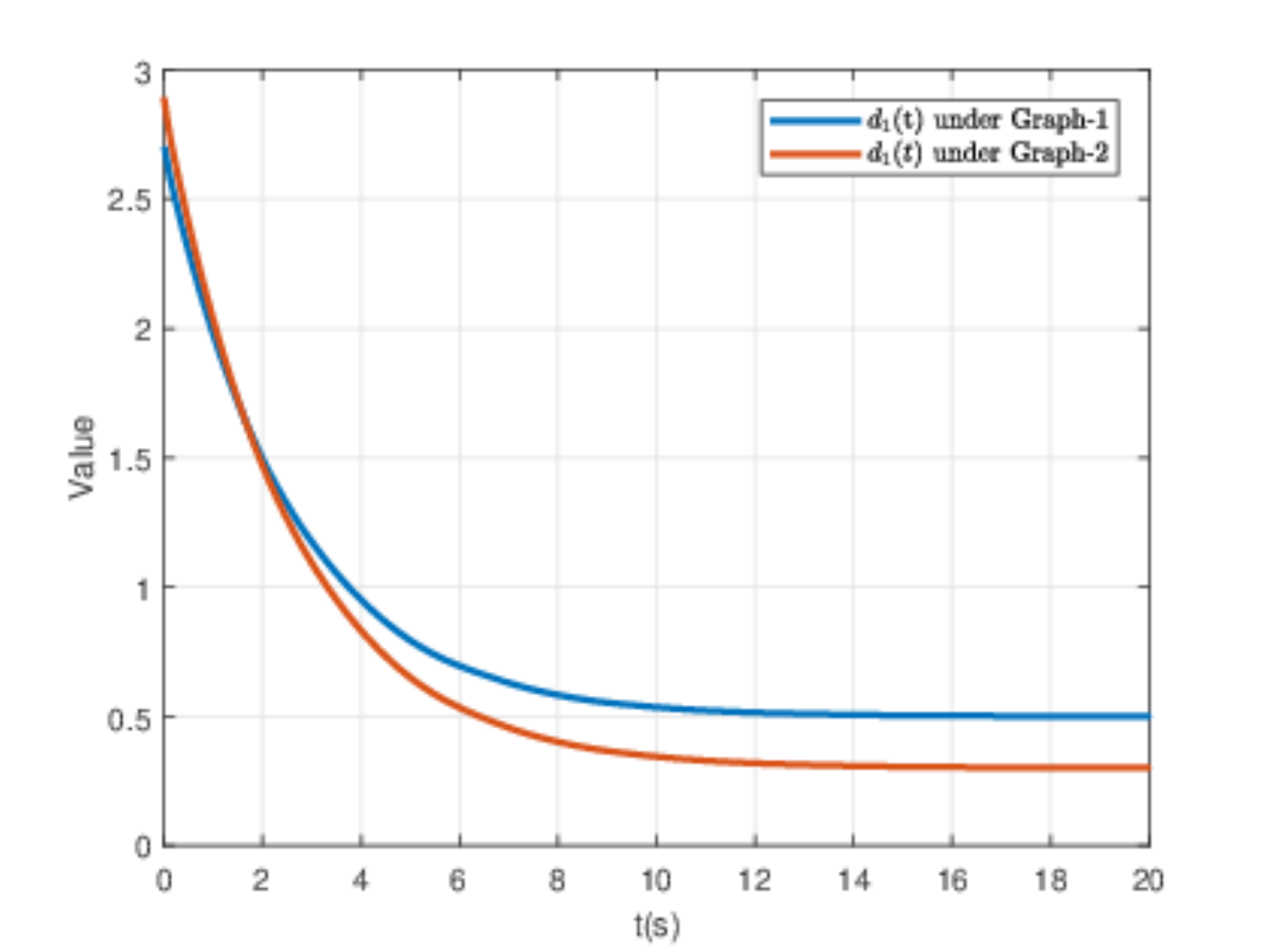}}%
	\caption{The time-varying threshold $d_1(t)$ of agent-1 under $\mathcal{G}_1$ and $\mathcal{G}_1$}
	\label{thre-gains-in2}
\end{figure}

\section{Conclusion}\label{con}

In this paper, we have proposed a novel event-triggered control protocol for leader-following consensus in second-order MASs under undirected topologies. 
To overcome the drawbacks of using continuous communicating signals among the follower agents, we have presented the event-triggered consensus protocol with an event-triggered mechanism. 
To get rid of the centralized information depending on the spectrum of the Laplacian matrix, we have proposed the adaptive laws to update the coupling gains and event-triggered thresholds resulting in that these parameters are free of the centralized information.
Moreover, considering that the velocity measurement process may be noise prone, we only use relative positions among agents in the protocol design. 
Compared with some existing results, the protocol in this paper has been less conservative and has excluded Zeno behavior. It has been found that consensus can be achieved under the distributed coupling gains and the distributed thresholds only if the undirected network is connected, which nevertheless is a simple and natural condition.

\section*{Acknowledgment}
This work was funded by National Science Foundation of China (No.61973040), China Postdoctoral Science Foundation (No.2020M680445), Postdoctoral Science Foundation of Beijing Academy of Agriculture and Forestry Sciences of China (No.2020-ZZ-001).

\bibliography{9th}
\end{document}